\documentclass{amsart}
\usepackage{graphicx}
\usepackage{amsmath, amsfonts, amssymb, amsthm}
\usepackage{bm}
\usepackage{enumerate}
\usepackage{todonotes}
\usepackage{booktabs}
\usepackage[colorlinks]{hyperref}
\newtheorem{thm}{Theorem}[section]
\newtheorem{cor}[thm]{Corollary}
\newtheorem{lem}[thm]{Lemma}

\theoremstyle{definition}
\newtheorem{defn}[thm]{Definition}
\theoremstyle{remark}
\newtheorem{rem}[thm]{Remark}
\numberwithin{equation}{section}
\newcommand{\RR}{\mathbb{R}}
\newcommand{\QQ}{\mathbb{Q}}
\newcommand{\PP}{\mathbb{P}}
\newcommand{\NN}{\mathbb{N}}

\newcommand{\cD}{\mathcal{D}}
\newcommand{\cE}{\mathcal{E}}
\newcommand{\cF}{\mathcal{F}}

\newcommand{\pp}{\textup{\texttt{+}}}
\newcommand{\mm}{\textup{\texttt{-}}}

\newcommand{\set}[1]{\left\{#1\right\}}

\newcommand{\Ind}[1]{\mathbf{1}_{\left\{#1\right\}}}

\newcommand{\Excond}[3]{\mathbb{E}^{#1}\left[\left.#2\right|#3\right]}  

\newcommand{\sseq}{\mathsf{sign}\textrm{-}\mathsf{seq}}

\newcommand{\sise}[1]{\bm{[}#1\bm{]}}
\newcommand{\head}{\overset{\scriptscriptstyle{H}}{\subseteq}}
\newcommand{\tail}{\overset{\scriptscriptstyle{T}}{\subseteq}}
\newcommand{\normal}{\texttt{normal}}
\newcommand{\humped}{\texttt{humped}}
\newcommand{\inverse}{\texttt{inverse}}
\newcommand{\dip}{\texttt{dipped}}
\newcommand{\hudi}{\texttt{HD}}
\newcommand{\dihu}{\texttt{DH}}
\newcommand{\HDH}{\texttt{HDH}}
\newcommand{\DHD}{\texttt{DHD}}
\newcommand{\HDHD}{\texttt{HDHD}}

\newcounter{mycounter}

\begin{document}

\title[Term structure shapes in the two-factor Vasicek model]{The classification of term structure shapes in the two-factor Vasicek model -- a total positivity approach}%
\author{Martin Keller-Ressel}%
\address{Institute for Mathematical Stochastics, TU Dresden}%
\email{martin.keller-ressel@tu-dresden.de}
\date{\today}

\begin{abstract}
We provide a full classification of all attainable term structure shapes in the two-factor Vasicek model of interest rates. In particular, we show that the shapes normal, inverse, humped, dipped and hump-dip are always attainable. In certain parameter regimes up to four additional shapes can be produced. Our results apply to both forward and yield curves and show that the correlation and the difference in mean-reversion speeds of the two factor processes play a key role in determining the scope of attainable shapes. The key mathematical tool is the theory of total positivity, pioneered by Samuel Karlin and others in the 1950ies.
\end{abstract}
\maketitle

\tableofcontents

\section{Introduction}
The term structure of interest rates -- summarized in the form of the yield or forward curve -- is one of the most fundamental economic indicators. Its shape encodes important information on the preferences for short- vs. long-term investments, the desire for liquidity and on expectations of central bank decisions and the general economic outlook. It is therefore a natural question, which shapes of yield and forward curves a given mathematical model of interest rates is able to (re-)produce. 
Already in \cite{vasicek1977equilibrium} a paragraph is dedicated to this question, with Vasicek concluding that normal (increasing), inverse (decreasing) and humped (endowed with a single maximum) shapes can be attained in his single-factor model. 
The same classification of shapes has been shown to hold in the Cox-Ingersoll-Ross model and furthermore in all one-dimensional affine term structure models (including short-rate models with jumps), see \cite[Eq.~(26)f]{cox1985theory}, \cite{keller-ressel2008yield, keller-ressel2018correction}.\\
For time-homogeneous \emph{multi-factor} models (such as the affine term structure models of \cite{dai2000specification}) there seems to be very little systematic knowledge on attainable term structure shapes. A notable exception is \cite{diez2019yield}, where it has been shown that the two-factor Vasicek model can also produce dipped curves, but without giving a complete enumeration of all other attainable shapes.\\
For time-inhomogeneous models, such as the Hull-White extended Vasicek model \cite{hull1990pricing}, it is well-known that \emph{any} initial term structure can be perfectly fitted and therefore that any shape of the term structure can be reproduced at the time of calibration. However, as time progresses, this initial shape will disappear and -- due to ergodicity effects -- the model will behave more and more like a time-homogeneous model. Therefore, even in view of Hull-White-extended models, the classification of attainable term structure shapes in time-homogeneous short-rate models remains a relevant question.\\
Here, we provide for the first time a full classification of term structure shapes in the two-factor Vasicek model. In our main result, Theorem~\ref{thm:main}, we classify all attainable shapes for both yield and forward curves.  As expected, several additional shapes, such as a dipped and a hump-dip curve, which are not attainable in the one-dimensional case, become attainable in the two-factor model. We strengthen and extend this main result in several ways: For many of the term structure shapes we can identify the exact region of the model's state space in which they occur. Moreover, we discuss which shapes are guaranteed to occur with strictly positive probability (`strict attainability') and for which shapes the locations of extrema can be arbitrarily prescribed (`strong attainability').

Our main mathematical tool is the theory of total positivity (see e.g. \cite{karlin1968total}), a theory linked to the variation-diminishing properties of certain matrices, function systems and integral kernels. Total positivity has broad applications in numerical interpolation, differential equations and stochastic processes. Within mathematical finance, it has been applied to study monotonicity and convexity of options prices \cite{kijima2002monotonicity} and to the principal-component-analysis of the term structure of interest rates \cite{salinelli2006correlation, lord2007level}. Our application to the shape analysis of the term structure is new and fundamentally different from the results in \cite{salinelli2006correlation, lord2007level}. While the results in this paper are limited to the two-dimensional Vasicek model, we are confident that the underlying theory can be applied to other multi-factor interest rate models as well.

\section{Notation and main result}
\subsection{Shapes of the term structure}

In our terminology \emph{term structure} refers to either the yield curve or the forward curve. The \emph{shape} $\mathsf{S}$ of the term structure is defined by the number and sequence of local maxima or minima of the term structure curve. In common financial market terminology a local maximum is called a `hump' and a local minimum a `dip'. As the term structure curves produced by the Vasicek model (or most other models) are smooth, it is clear that the shape of the term structure curve can be conveniently analyzed by considering its derivative: Any sign change of the derivative (from strictly positive to strictly negative or vice versa) corresponds to a local extremum of the term structure; the type of sign change ($\pp$ to $\mm$ or $\mm$ to $\pp$) determines the type of the extremum (hump or dip). The basic shapes and their conventional names are listed in Table~\ref{tab:shape}. For `higher order' shapes we use the letters \texttt{H} for a hump and \texttt{D} for a dip, e.g., the shape $\HDH{}$ corresponds to a term structure with two local maxima, interlaced by a single local minimum.
\setlength{\tabcolsep}{8pt}
\begin{table}[hbtp]\label{tab:shape}
\begin{center}
\begin{tabular}{p{3cm}p{5cm}p{3cm}} 
\toprule
Shape $\mathsf{S}$ of the term structure & Description & Sign sequence of derivative\\ 
\midrule 
\normal{} & strictly increasing & $\sise{\pp}$\\
\inverse{} & strictly descreasing & $\sise{\mm}$\\
\humped{} & single local maximum & $\sise{\pp \mm}$\\
\dip{} & single local minimum & $\sise{\mm \pp}$\\
\hudi{} & hump-dip, i.e. local maximum followed by local minimum & $\sise{\pp \mm \pp}$\\
\dihu{}, \HDH{}, etc. & further sequences of multiple `dips'  and `humps'  & $\sise{\dotsc}$ \\
\bottomrule
\end{tabular}
\end{center}
\caption{Shapes of the term structure}
\end{table}

\subsection{The two-factor Vasicek model} The Vasicek model, originally introduced by \cite{vasicek1977equilibrium} as a single-factor model, has been extended to multiple factors by \cite{dai2000specification} within the framework of affine term structure models. Here, we focus on the two-dimensional case, which has been treated in detail e.g. in \cite{brigo2007interest}. In the two-dimensional Vasicek model the short rate is given by
\[r_t = Z_t^1 + Z_t^2,\]
where the dynamics of the factor process $Z = (Z^1,Z^2)$ are given by 
\begin{align}\label{eq:vasicek}
dZ^i_t = -\lambda_i ( Z_t^i - \theta_i) \, dt + \sigma_i dB_t^i, \qquad i \in \set{1, 2}.
\end{align}
under the risk-neutral measure $\QQ$. 
The long-term rates $\theta = (\theta_1, \theta_2)$ are real and the Brownian motions $B^1, B^2$ have correlation $\rho \in [-1,1]$. We assume that the mean-reversion speeds are strictly positive and ordered as
\[\lambda_1 < \lambda_2,\]
i.e. $Z^1$ is the `slow' factor with predominate influence on the long end of the term structure, and $Z^2$ the `fast' factor with predominate influence on short-term rates.  
From \cite{dai2000specification}, the bond price in the two-dimensional Vasicek model can be written as
\begin{equation}\label{eq:bond}
P(t,t+x) = \Excond{\QQ}{\exp\left(- \int_t^{t+x} r_s ds\right)}{\cF_t} = \exp \left(A(x) + Z_t^\top B(x)\right) 
\end{equation}
where $A$ and $B$ are given as solutions of the ODEs
\begin{subequations}\label{eq:Riccati}
\begin{align}
A'(x) &= F(B(x)), \qquad A(0) = 0\label{eq:Riccati_A}\\
B'_i(x) &= R_i(B_i(s)), \qquad B_i(0) = 0, \qquad i \in \set{1,2}\label{eq:Riccati_B}
\end{align}
\end{subequations}
with 
\begin{subequations}\label{eq:FR_def}
\begin{align}
F(b) &= \lambda_1 \theta_1 b_1 + \lambda_2 \theta_2 b_2 + \frac{1}{2} b^\top \Sigma b, \qquad \quad \Sigma = \begin{pmatrix} \sigma_1^2 & \rho \sigma_1 \sigma_2 \\ \rho \sigma_1 \sigma_2 & \sigma_2^2 \end{pmatrix},\label{eq:F_def}\\
R_i(b) &= -\lambda_i b_i - 1, \qquad i \in \set{1,2}.\label{eq:R_def}
\end{align}
\end{subequations}
The differential equations \eqref{eq:Riccati_B} can obviously be solved explicitly with solutions given by
\[B_i(x) = \frac{1}{\lambda_i} \left(e^{-\lambda_i x} - 1\right), \quad i \in \set{1,2}.\]
By integration, the explicit solution of \eqref{eq:Riccati_A} is given by
\begin{align*}
A(x) &= \frac{\sigma_1^2}{4 \lambda_1^3} \left(e^{-2\lambda_1 x} + 4 e^{-\lambda_1 x} - 2\lambda_1 x + 3\right) +  \frac{\sigma_2^2}{4 \lambda_2^3} \left(e^{-2\lambda_2 x} + 4 e^{-\lambda_2 x} - 2\lambda_2 x + 3\right) \\
&+ \frac{\sigma_1 \sigma_2}{\lambda_1 \lambda_2}\left(\frac{e^{-(\lambda_1 + \lambda_2) x}}{\lambda_1 + \lambda_2} - \frac{e^{-\lambda_1 x}}{\lambda_1} -\frac{e^{-\lambda_2 x}}{\lambda_2} - x\right).
\end{align*}

Finally, the yield and forward curves in the Vasicek model are easily computed from \eqref{eq:bond} and \eqref{eq:Riccati} as
\begin{align}
f(x;Z_t) &= -\partial_x \log P(t,t+x) = -A'(x) - Z_t^\top B'(x)\label{eq:forward},\\
Y(x;Z_t) &= -\frac{1}{x} \log P(t,t+x) = - \frac{A(x)}{x} - Z_t^\top \frac{B(x)}{x}. \label{eq:yield}
\end{align}
If we want to emphasize the dependency of these curves on some parameter $p$ in addition to the state vector $z$, we write $f(x;z,p)$ and $Y(x;z,p)$. We use the analogous notation for all quantities derived from $f$ and $Y$.\\

\subsection{Classification of term structure shapes}
We are now prepared to present the main result of this paper; the classification of term structure shapes in the two-factor Vasicek model. We denote by $\bm{P}$ the full parameter space of the two-dimensional Vasicek model, i.e. 
\[\bm{P} = \set{\begin{pmatrix}\theta_1\\\theta_2\end{pmatrix} \in \RR^2, \begin{pmatrix}\sigma_1\\\sigma_2\end{pmatrix} \in [0,\infty)^2, \rho \in [-1,1], 0 < \lambda_1 < \lambda_2},\]
and introduce the following definition:
\begin{defn}[Attainability]\label{def:attainable}
A shape $\mathsf{S}$ of the forward curve is called \textbf{attainable}, if we can find a parameter vector $p \in \bm{P}$ and a state vector $z \in \RR^2$, such that $x \mapsto f(x; z, p)$ has shape $\mathsf{S}$. The same definition applies to the yield curve $x \mapsto Y(x; z, p)$.
\end{defn}
Moreover, the relation of the two speed-of-mean-reversion parameters $\lambda_1 < \lambda_2$, is distinguished as follows:
\begin{defn}The two-dimensional Vasicek model is called 
\begin{itemize}
\item \textbf{scale-separated}, if $2 \lambda_1 < \lambda_2$,
\item \textbf{scale-proximal}, if $2 \lambda_1 > \lambda_2$, and
\item \textbf{scale-critical}, if $2 \lambda_1 = \lambda_2$.
\end{itemize}
\end{defn}

\begin{thm}\label{thm:main}
Consider the two-dimensional Vasicek model.
\begin{enumerate}[(a)]
\item In the scale-separated case, the following yield and forward curve shapes are attainable: 
\begin{quote}
\normal{}, \inverse{}, \humped{}, \dip{}, \hudi{}, \dihu{},  \HDH{};
\end{quote}
no other shapes are attainable.
\item In the scale-proximal case, the following yield and forward curve shapes are attainable with $\rho \ge 0$:
\begin{quote}
\normal{}, \inverse{}, \humped{}, \dip{}, \hudi{};
\end{quote}
no other shapes are attainable with $\rho \ge 0$.
\item In the scale-proximal case, the following yield and forward curve shapes are attainable with $\rho < 0$:
\begin{quote}
\normal{}, \inverse{}, \humped{}, \dip{}, \hudi{}, \dihu{},  \HDH{}, \DHD{}, \HDHD{};
\end{quote}
no other shapes are attainable with $\rho < 0$.
\end{enumerate}
In the scale-critical case, (b) applies if $\rho \ge 0$ and (a) applies if $\rho < 0$. 
\end{thm}
We observe hat compared to the one-dimensional Vasicek model, in which only the three shapes \normal{}, \inverse{}, and \humped{} can be produced (see \cite{vasicek1977equilibrium, keller-ressel2008yield, korn2013optionsbewertung}), at least the two additional shapes \dip{} and \hudi{} are attainable in the two-factor case. Depending on the relation of the speed-of-mean-reversion parameters and on the correlation $\rho$ the number of additional shapes can grow up to six. 
The proof of the theorem is is given in section~\ref{sec:proof}. It is based on the theory of total positivity and Descartes systems, which is summarized in section~\ref{sec:totpos}).
In section~\ref{sec:additional}, the main result is refined and extended in several ways: Firstly, the analysis of attainable term structure shapes can also be carried out \emph{contingent on the state vector $(Z_t^1, Z_t^2)$}. This allows to partition the state space into regions in which only a few or even a single shape is possible; see Figure~\ref{fig:1} for an illustration of the scale-proximal and positively correlated case. Secondly, we introduce and discuss the notions of \emph{strict} and \emph{strong} attainability, which essentially correspond to the attainability of shapes with strictly positive probability and with arbitrary placement of local extrema. Finally, we show in section~\ref{sec:additional} that all attainable shapes can be produced by just varying the state vector and the volatility parameters $\sigma_1 ,\sigma_2, \rho$, while keeping all other model parameters fixed.

\begin{figure}[!p]
\centering
\includegraphics[width=0.65\textwidth]{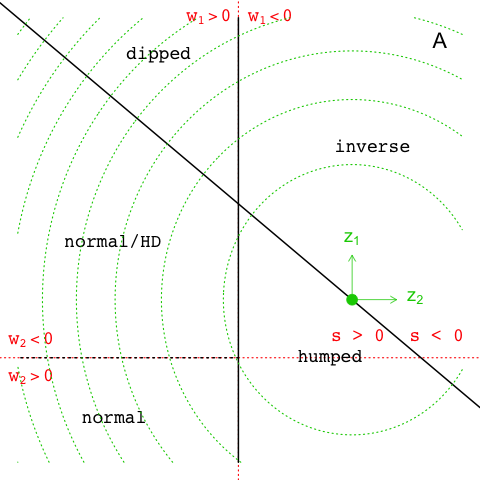}
\vspace{2em}
\includegraphics[width=0.65\textwidth]{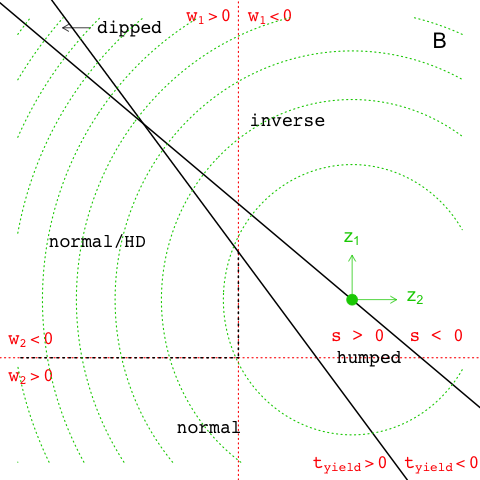}
\caption{\textbf{State-contingent term structure shapes in the two-dimensional Vasicek model.} The attainable shapes of the forward curve (panel A) and the yield curve (panel B) in different regions (delimited in black) of the state space $\RR^2$ of the factors $(z_1,z_2)$ are shown. The underlying model is assumed to be scale-proximal and positively correlated, see case (b) of Theorem~\ref{thm:main}.  The contour lines of the risk-neutral stationary distribution of the factor process are shown in green, and the green dot shows the location of the `most likely' yield/forward curve. See section~\ref{sec:state} for details on the annotation in red.}
\label{fig:1}
\end{figure}

\section{Sign sequences, total positivity and Descartes systems}\label{sec:totpos}

\subsection{Sign sequences}
In order  to keep track of the number and the directions of sign changes of a numeric sequence or of a continuous function we introduce the notion of a \emph{sign sequence}. While this notion appears implicitly in many of the results related to total positivity, the exact terminology and notation introduced here is new.
\begin{enumerate}[(i)]
\item A \textbf{sign sequence} is a non-empty sequence of the symbols $\pp$ and $\mm$. Only finite sign sequences will be considered here. Also zeroes can be allowed; we comment on this later. We include sign sequences in square brackets and write e.g.
\[\sise{\pp}, \qquad \sise{\pp\pp\mm\mm\pp}, \qquad \sise{\pp\mm\pp}\]
for some valid sign sequences. 
\item Two sign sequences are \textbf{equivalent}, if the number and direction of their sign changes is the same. This defines an equivalence relation $\simeq$, e.g., 
\[\sise{\pp\pp\mm\mm\pp} \simeq \sise{\pp\mm\pp\pp\pp}.\]
\item In a similar way we can define a \textbf{subsequence} relation $\subseteq$ in which only the sign changes are considered (i.e. we treat blocks of signs as if they were single signs). Thus, we have
\[\sise{\mm\mm\pp\pp\pp} \subseteq \sise{\mm\pp\mm\mm}, \qquad \sise{\mm} \subseteq \sise{\pp\pp\mm\pp}.\]
\item A subsequence which also preserves the initial sign is called a \textbf{head} and a subsequence which preserves the terminal sign is called a \textbf{tail}. We write
\[\sise{\mm\mm\pp\pp\pp}\head \sise{\mm\pp\mm\mm}, \qquad \sise{\pp}\tail \sise{\mm\mm\pp}\]
for the respective relations.
\item Sign sequences should only keep track of `strong' sign changes.\footnote{A sign change from $\pp$ to $0$ and back to $\pp$, for example, is not considered a strong sign change, whereas a sign change from $\pp$ to $0$ and then to $\mm$ is.} Therefore we add the convention that \textbf{zeroes} in sign sequences can simply be omitted to obtain an equivalent sign sequence. E.g. we have
\[\sise{\pp0\pp\pp\mm0\pp} \simeq \sise{\pp\mm\pp}, \qquad \sise{0\mm00\mm} \simeq \sise{\mm}.\]
Note that all strong sign changes (and their direction) are preserved under this reduction.
\item If a variable, say $a$, appears inside a sign sequence, it should be interpreted as `sign of $a$'. E.g. the sign sequence $\sise{ab}$ evaluates to $\sise{\pp\mm}$ if $a = 6$ and $b = -1$ and to $\sise{\mm}$ if $a = -1$, $b = 0$.
\item Let $f$ be a continuous function, defined on a subset $X$ of $\RR$ and not constantly zero. The \textbf{sign sequence of $f$} is the sequence of signs that $f$ takes on between its zeroes. Only functions with finite sign sequences will be considered and we denote the sign sequence of such a function $f$ by $\sseq(f)$. For example
\[f(x) = x^2 -1\,\text{, defined on $X = [0,\infty)$} \qquad \Longrightarrow \qquad \sseq(f) = \sise{\mm\pp}.\]
\end{enumerate}

\subsection{Total positivity and Descartes systems}
\label{sec:tot_pos}
We introduce some definitions and key results from the theory of total positivity. For background and further details we refer to \cite{karlin1966tchebycheff, karlin1968total} and \cite{borwein1995polynomials}.

 \begin{defn}[Totally positive kernel] Let $X, Y \subseteq \RR$ and let $K$ be a function (`kernel') from $X \times Y$ to $\RR$. If 
\begin{equation}\label{eq:K}
\det \begin{pmatrix} K(x_1,y_1) & K(x_1, y_2) & \dotsc &  K(x_1, y_m) \\ \vdots & \vdots && \vdots\\  K(x_m, y_1) &  K(x_m,y_2) & \dotsc &  K(x_m,y_m)\end{pmatrix} \ge 0
\end{equation}
for any $m \in \NN$, $x_1 < x_2 < \dotsm < x_m$ in $X$ and $i_1 < i_2  < \dotsm < i_m$ in $Y$, then $K(x,y)$ is called \textbf{totally positive}. If strict equality holds in \eqref{eq:K}, the kernel is called \textbf{strictly totally positive}.
\end{defn}
\begin{rem}
\begin{enumerate}[(i)]
\item The kernels $K(x,y) = e^{xy}$ and $K(x,y) = \Ind{y \le x}$ are examples of totally positive kernels on $\RR^2$ (or on any $X \times Y$ with $X, Y \subseteq \RR$); see \cite[Ch.1 \S 2]{karlin1968total} and \cite[Ch.3, Eq.~(1.13)ff]{karlin1968total}. The first kernel is even strictly totally positive.
\item A totally positive kernel on $X=Y= \set{1, \dotsc, n}$ can be written as a matrix; accordingly such matrices are also called totally positive, cf. \cite{ando1987totally} or \cite[Ch.29]{hogben2013handbook}.
\end{enumerate}
\end{rem}
A crucial property of totally positive kernels is the following:
 \begin{thm}[Variation-diminishing property of totally positive kernels]\label{thm:vardim_K}
Let $K$ be a totally positive kernel on $X \times Y$, such that $\int_Y K(x,y) dy < \infty$ for all $x \in X$. Let $f: Y \to \RR$ be a bounded continuous function with finite sign sequence and set
\[g(x) := \int_Y K(x,y) f(y) dy.\]
Then 
 \[\sseq(g)  \subseteq \sseq(f).\]
 \end{thm}
This result is a particular case of \cite[Ch.~5, Thm.~3.1]{karlin1968total}, formulated in the language of sign sequences. It can be extended from integration with respect to Lebesgue measure $dy$ to a large class of $\sigma$-finite measures $d\mu(y)$ on $Y$. These extensions, however, will not be needed here.\\
 
 Next, we discuss a closely related definition, which applies to families of functions. 

\begin{defn}[Descartes system] Let $X$ be a subinterval of $\RR$ and let $\cD = (\phi_1, \dotsc, \phi_n)$ be a family of continuous functions from $X$ to $\RR$. If
\begin{equation}\label{eq:Phi}
\det \begin{pmatrix} \phi_{i_1}(x_1) & \phi_{i_2}(x_1) & \dotsc &  \phi_{i_m}(x_1)\\ \vdots & \vdots && \vdots\\  \phi_{i_1}(x_m) &  \phi_{i_2}(x_m) & \dotsc &  \phi_{i_m}(x_m)\end{pmatrix} > 0
\end{equation}
for any $m \le n$, $x_1 < x_2 < \dotsm < x_m$ in $X$ and $i_1 < i_2  < \dotsm < i_m$ in $\set{1, \dotsc, n}$, then 
$\cD$ is called a \textbf{Descartes system} on $X$.
\end{defn}
\begin{rem}\label{rem:Descartes}
\begin{enumerate}[(i)]
\item The order of the functions $\phi_1, \dotsc, \phi_n$ matters and a permutation of a Descartes system need not be a Descartes system.
\item A Descartes system can be seen as a strictly totally positive kernel on $X \times \set{1, \dotsc, n}$
\item The family of monomials $(1,x,x^2,x^3, \dotsc, x^n)$ is a Descartes system.
\item The family of exponential functions $(e^{x\gamma_1}, \dotsc, e^{x \gamma_n})$ is a Descartes system if and only if $\gamma_1 < \gamma_2 < \dotsm < \gamma_n$ \label{item:exp}
\end{enumerate}
\end{rem}
Also Descartes systems enjoy variation-diminishing properties:
\begin{thm}[Variation-diminishing property of Descartes systems]\label{thm:vardim}
Let $(\phi_1, \dots, \phi_n)$ be a Descartes system and let $(a_1, \dotsc, a_n) \in \RR^n$. Then
\begin{equation}\label{eq:vardim}
\sseq\left(\sum_{i=1}^n a_i \phi_i\right) \subseteq \sise{a_1 a_2, \dotsm a_n}.
\end{equation}
\end{thm}
\begin{rem}
\begin{enumerate}[(i)]
\item This theorem is \cite[Thm.~3.1, 4.4]{karlin1966tchebycheff} (see also \cite[Thm.~3.2.4]{borwein1995polynomials}), translated into the language of sign sequences.
\item The well-known \emph{Descartes' rule of signs}  for polynomials follows by applying this theorem to the Descartes system $(1, x, x^2, \dotsc, x^n)$; see \cite[3.2.E7]{borwein1995polynomials}.
\end{enumerate}
\end{rem}
Given a Descartes system $\cD = (\phi_1, \dots, \phi_n)$, a function of the form 
\[\phi(x) := \sum_{i=1}^n a_i \phi_i(x)\]
is called a \textbf{D-polynomial} in $\cD$. We call $\phi$ \textbf{extremal}, if equality is attained in \eqref{eq:vardim}. The next result concerns the interpolation properties of D-polynomials:
\begin{thm}\label{thm:extremal}
Let $(\phi_1, \dots, \phi_n)$ be a Descartes system on $X$ and let $r_1 < r_2 < \dotsm < r_{n-1}$ be $n-1$ distinct points in $X$. Then there exists a $D$-polynomial $\phi(x) = \sum_{i =1}^n a_i \phi_i(x)$ with all $a_i$ non-zero, which satisfies:
\begin{itemize}
\item $\phi(r_i) = 0$ for all $i \in 1, \dotsc, n-1$;
\item $\phi$ has a strong sign change at each $r_i$ in the interior of $X$.
\end{itemize}
If all $r_i$ are interior points of $X$, then $\phi$ is extremal, i.e., 
\begin{itemize}
\item $\sseq(\phi) \simeq \sise{a_1\,a_2\dotsm a_n}$.
\end{itemize}
\end{thm}
This result follows from \cite[Ch.~I, Thm.~5.1]{karlin1966tchebycheff} or \cite[3.1.E11]{borwein1995polynomials}, but we provide a self-contained proof and some related results in Sec.~\ref{app:interpolation} and \ref{app:interpolation2}.

\section{Proof of the main result}\label{sec:proof}
The proof of Theorem~\ref{thm:main} and its corollaries rests on identifying Descartes systems related to yield and forward curves in the two-dimensional Vasicek model. These Descartes systems are given in Section~\ref{sec:Descartes} below and allow to apply the results from the theory of total positivity from above. The proof of Theorem~\ref{thm:main} is then given in two parts: First, in Section~\ref{sec:admissible}, we show necessity, i.e., that no term structure shapes outside of the lists given in Theorem~\ref{thm:main} can be attained. Then we show sufficiency, i.e., that all listed shapes are actually attainable. This more difficult part is done in Section~\ref{sec:attainable}.

\subsection{Descartes systems for the Vasicek model}\label{sec:Descartes}
We introduce several Descartes systems associated to the two-dimensional Vasicek model. As we will show, the derivatives of the forward curve and the yield curve can be written as D-polynomials in these systems. The next Lemma follows directly from Remark~\ref{rem:Descartes}(\ref{item:exp}) and from the ordering of exponents that is implied by the scale-separation properties: 
\begin{lem}\label{lem:Descartes_f} The following families of functions are Descartes systems on $[0,\infty)$: 
\begin{align*}
\cD_\text{sep} &= (e^{-2 \lambda_2 x}, e^{-(\lambda_1 + \lambda_2)x}, e^{-\lambda_2 x}, e^{-2 \lambda_1 x},  e^{-\lambda_1 x}) \quad &&\text{if} \quad  2 \lambda_1 < \lambda_2\\
\cD_\text{prox} &= (e^{-2 \lambda_2 x}, e^{-(\lambda_1 + \lambda_2)x}, e^{-2 \lambda_1 x}, e^{-\lambda_2 x},  e^{-\lambda_1 x}) \quad &&\text{if} \quad  2 \lambda_1 > \lambda_2\\
\cD_\text{crit} &= (e^{-2 \lambda_2 x}, e^{-(\lambda_1 + \lambda_2)x}, e^{-\lambda_2 x}, e^{-\lambda_1 x}) \quad &&\text{if} \quad  2 \lambda_1 = \lambda_2
\end{align*}
\end{lem}
Note that the only difference between $\cD_\text{prox}$ and $\cD_\text{sep}$ are the order of the third and the fourth element. Collapsing these cases yields the boundary case $\cD_\text{crit}$.\\ 

For the analysis of yield curve shapes a slightly different Descartes system is needed:
\begin{lem} \label{lem:Descartes_g}
Set
\begin{equation}\label{eq:g_def}
g_\alpha(x) = \frac{1}{x^2}\int_0^x y e^{-\alpha y} dy = \tfrac{1}{ \alpha^2 x^2} \left(e^{-\alpha x} - 1 + \alpha x e^{-\alpha x}\right).
\end{equation}
The following families of functions are Descartes systems on $[0,\infty)$:
\begin{align*}
\cE_\text{sep} &= \left(g_{2 \lambda_2}, g_{\lambda_1 + \lambda_2},  g_{\lambda_2}, g_{2 \lambda _1}, g_{\lambda_1}\right) \quad &&\text{if} \quad 2 \lambda_1 < \lambda_2\\
\cE_\text{prox} &= \left(g_{2 \lambda_2}, g_{\lambda_1 + \lambda_2}, g_{2 \lambda _1}, g_{\lambda_2}, g_{\lambda_1}\right) \quad &&\text{if} \quad 2 \lambda_1 > \lambda_2\\
\cE_\text{crit} &= \left(g_{2 \lambda_2}, g_{\lambda_1 + \lambda_2}, g_{\lambda_2}, g_{\lambda_1}\right) \quad &&\text{if} \quad 2 \lambda_1 = \lambda_2
\end{align*}
\end{lem}
Note that $g_\alpha(x)$ can be written as
\begin{equation}\label{eq:g_rep}
g_\alpha(x) = \int_0^\infty K(x,y) e^{-\alpha y} dy, \quad \text{where} \quad K(x,y) = \frac{y}{x^2}\Ind{y \le x}.
\end{equation}
The introduced kernel is of the form $K(x,y) = \phi(x)\psi(y)L(x,y)$, where $\phi(x) = \frac{1}{x^2}$, $\psi(y) = y$ are strictly positive on $(0,\infty)$ and where $L(x,y) = \Ind{y \le x}$. The total positivity of $L(x,y) = \Ind{y \le x}$ is shown in \cite[Ch.~3, Eq.(1.10)ff]{karlin1968total}, and the total positivity of the composed kernel $K(x,y)$ follows from \cite[Ch.~1, Thm.~2.1]{karlin1968total}. Thus, the systems $\cE_{(\dots)}$ are \emph{totally positive transformations} of the systems $\cD_{(\dots)}$. This immediately implies that they are `weak Descartes systems' on $(0,\infty)$, i.e. that \eqref{eq:Phi} holds with non-strict inequality. A full proof of their `strong' Descartes property including the boundary point $x=0$ is given in Section~\ref{app:Descartes_E}.

\subsection{Necessary conditions for attainability}\label{sec:admissible}

To derive necessary conditions for attainability of term structure shapes, we write the derivatives of the forward and the yield curve as D-polynomials in the Descartes systems introduced in Lemmas~\ref{lem:Descartes_f} and \ref{lem:Descartes_g} and determine their coefficients. The first step is to calculate the derivatives:

\begin{lem}\label{lem:curve_diff}The derivative of the forward curve in the Vasicek model is given by 
\begin{equation}\label{eq:l_poly}
\partial_x f(x) = u_2 \varphi_{2\lambda_2}(x) + c \varphi_{\lambda_1 + \lambda_2}(x) + w_2 \varphi_{\lambda_2}(x) +  u_1 \varphi_{2\lambda_1}(x) + w_1 \varphi_{\lambda_1}(x), 
\end{equation}
with $\varphi_\alpha(x) = e^{-\alpha x}$ and coefficients given, for $j \in \set{1,2}$,  by 
\begin{align*}
u_j &= \frac{\sigma_j^2}{\lambda_j} \ge 0\\
w_j &= w_j(z_j) = \lambda_j \left(\theta_j - z_j\right) - \frac{\sigma_j^2}{\lambda_j} - \rho \lambda_j \frac{\sigma_1\sigma_2}{\lambda_1\lambda_2}
\intertext{and}
c &= \rho (\lambda_1 + \lambda_2) \frac{\sigma_1\sigma_2}{\lambda_1\lambda_2}.
\end{align*}
The derivative of the yield curve is given by 
\begin{align}\label{eq:m_poly}
\partial_x Y(x) &= \frac{1}{x^2} \int_0^x \partial_xf(y) \, y dy = \notag \\
&= u_2 g_{2\lambda_2}(x) + c g_{\lambda_1 + \lambda_2}(x) + w_2 g_{\lambda_2}(x) +  u_1 g_{2\lambda_1}(x) + w_1 g_{\lambda_1}(x), 
\end{align}
with $g_\alpha(x)$ given by \eqref{eq:g_def}.
\end{lem}
\begin{proof}
From \eqref{eq:forward} we obtain
\[\partial_x f(x;z) = -A''(x)  - z^\top B''(x),\]
which evaluates to 
\begin{align*}
\partial_x f(x;z) &=  \begin{pmatrix} e^{-\lambda_1 x} \\e^{-\lambda_2 x} \end{pmatrix}^\top 
\left\{ - \begin{pmatrix}\lambda_1 \theta_1\\\lambda_2 \theta_2 \end{pmatrix} -  \begin{pmatrix}\sigma_1^2 & \rho \sigma_1 \sigma_2\\\rho \sigma_1 \sigma_2 &  \sigma_2^2\end{pmatrix}  \begin{pmatrix}\frac{1}{\lambda_1} (e^{-\lambda_1 x} - 1)\\\frac{1}{\lambda_2} (e^{-\lambda_2 x} - 1)\end{pmatrix} + \begin{pmatrix}z_1\\z_2 \end{pmatrix} \right\}
\end{align*}
and after rearrangement gives \eqref{eq:l_poly}.\\
For the yield curve, differentiation of \eqref{eq:yield} gives
\[\partial_x Y(x;z) = \frac{1}{x^2} \left\{\left(A(x) - xF(B(x))\right) + z^\top \left(B(x) - xR(B(x))\right)\right\}.\]
Multiplying with $x^2$ and taking another derivative we obtain
\[\partial_x(x^2 \partial_x Y(x;z)) = - x A''(x)   - x z^\top B''(x) = x \partial_x f(x;z),\]
which yields \eqref{eq:m_poly}.
\end{proof}

Combining this result with Lemmas~\ref{lem:Descartes_f} and \ref{lem:Descartes_g}, we obtain the following: 

\begin{lem}\label{lem:coef}The functions $\partial_x f$ and $\partial_x Y$ are D-polynomials in the Descartes systems $\cD$ and $\cE$ respectively, with coefficients given by 
\begin{itemize}
\item $(u_2, c, u_1, w_2, w_1)$ in the scale-proximal case,
\item $(u_2, c, w_2, u_1, w_1)$ in the scale-separated case,
\item $(u_2, c, w_2 + u_1, w_1)$ in the scale-critical case.
\end{itemize}
\end{lem}

We can now use the variation-diminishing property of Descartes systems to derive restrictions on attainable forward and yield curve shapes. 

\begin{thm}\label{thm:sise_m}If $\rho \ge 0$, then the sign sequence of the derivative of the forward and the yield curve, $d \in \set{\partial_x f,\partial_x Y}$, satisfies
\begin{align*}
&\sseq(d)\subseteq \sise{\pp w_2 w_1} \quad &&\text{(under scale-proximity)}\\
&\sseq(d)\subseteq \sise{\pp w_2 \pp w_1} \quad &&\text{(under scale-separation)}\\
&\sseq(d)\subseteq \sise{\pp (u_1 + w_2) w_1} &&\quad \text{(under scale-criticality)}.
\end{align*}
If $\rho < 0$ then the sign sequence of $d \in \set{\partial_x f,\partial_x Y}$ satsifies
\begin{align*}
&\sseq(d)\subseteq \sise{\pp \mm \pp w_2 w_1} \quad &&\text{(under scale-proximity)}\\
&\sseq(d)\subseteq \sise{\pp \mm w_2 \pp w_1} \quad &&\text{(under scale-separation)}\\
&\sseq(d)\subseteq \sise{\pp \mm (u_1 + w_2) w_1} &&\quad \text{(under scale-criticality)}.
\end{align*}
\end{thm}
For forward curves this result can be strengthened by using additional information from the terminal sign of $\partial_x f$.
\begin{cor}\label{cor:sise_l}
In Theorem~\ref{thm:sise_m} `$\subseteq$' can be replaced by `$\tail$' whenever the sign sequence of $\partial_x f$ is considered.
\end{cor}
\begin{proof}
Theorem~\ref{thm:sise_m} follows by applying Theorem~\ref{thm:vardim} to the coefficients given in Lemma~\ref{lem:coef}. In doing so, we take into account that $u_j$ has positive sign regardless of the choice of parameters, and apply the reductions of sign sequences described in Sec.~1.1 to arrive at the expressions on the right hand sides.\\
For the corollary, the obtained relations can be strengthened from $\subseteq$ to $\tail$ by analyzing the terminal sign (the sign after the last sign change) of $\partial_x f$. From Lemma~\ref{lem:curve_diff} we obtain that $\lim_{x \to \infty} \partial_x f(x) = 0$, which, however, yields no information on the terminal sign. Rather, the terminal sign of $\partial_x f$ must be  determined by the component with the slowest decay, which is $w_1 \varphi_{\lambda_1}(x) = w_1 e^{-\lambda_1 x}$. Thus, the terminal sign of $\partial_x f$ is equal to the sign of $w_1$, which is the last sign in all sequences of Lemma~\ref{lem:coef}. We conclude that $\sseq(\partial_x f)$ is not just a subset, but rather a tail of all the sign sequences that were obtained on the right hand sides.\footnote{Note that the same approach does not work for $\partial_x Y$ due to the different asymptotic behaviour as $x$ tends to infinity.}
\end{proof}

Using Theorem~\ref{thm:sise_m} we obtain the first part of our main result, Theorem~\ref{thm:main}.

\begin{proof}[Proof of Theorem~\ref{thm:main} -- necessity]
Consider the case of the forward curve. The shape of the forward curve is determined by the sign sequence of $\partial_x f$, and this sign sequence is controlled by the results of Corollary~\ref{cor:sise_l}. Hence, restrictions on attainable term structure shapes can be obtained by iterating through all cases of Corollary~\ref{cor:sise_l} and through the four possible sign combinations of $w_1$ and $w_2$. Note that we only need to consider the strict signs $\pp$ and $\mm$, because zeroes can be omitted from sign sequences and do not lead to additional shapes. Instead of listing all possible combinations, we discuss two exemplary cases: 
\begin{itemize}
\item Suppose that $\rho \ge 0$, $w_1 > 0$ and $w_2 < 0$. In the scale-proximal case we obtain from Corollary~\ref{cor:sise_l}, that
\[\sseq(\partial_x f)\tail \sise{\pp \mm \pp}.\]
The possible tail sequences of $\sise{\pp \mm \pp}$ are $\sise{\pp}, \sise{\mm\pp}$ and $\sise{\pp \mm \pp}$ itself. These cases correspond to the shapes \normal{}, \humped{} and \hudi{}, and we conclude that no other forward curve shapes can be attainable under the given parameter restrictions. Switching to scale-separation, Corollary~\ref{cor:sise_l} yields
\[\sseq(\partial_x f)\tail \sise{\pp \mm \pp \pp} \simeq \sise{\pp \mm \pp},\]
and the same admissible shapes are obtained as in the scale-proximal case. 
\item Now suppose that $\rho \ge 0$, $w_1 < 0$ and $w_2 < 0$. In the scale-proximal case Corollary~\ref{cor:sise_l} yields
\[\sseq(\partial_x f)\tail \sise{\pp \mm \mm} \simeq \sise{\pp \mm},\]
which leaves the shapes \inverse{}, \humped{} as potentially attainable shapes. In the scale-separated case we obtain
\[\sseq(\partial_x f)\tail \sise{\pp \mm \pp \mm},\]
which, in addition, leaves \dihu{} and \HDH{} as potentially attainable.
\end{itemize}
Applying the same procedure to all other cases produces the lists given in the theorem, in the case of forward curves. The scale-critical case can be treated like the scale-proximal case if $\rho \ge 0$, and like the scale-separated case if $\rho < 0$. 
For yield curves, we apply Theorem~\ref{thm:sise_m} to $\partial_x Y$ in the same manner. Despite the weaker constraint $\subseteq$ instead of $\tail$, it turns out (after iterating through all cases) that the same lists of shapes are obtained. 
\end{proof}

\subsection{Sufficient conditions for attainability}\label{sec:attainable}
To complete the proof of Theorem~\ref{thm:main}, we need to show sufficiency, i.e., that all listed shapes are actually attainable. Before going into details, we describe the general strategy of the proof: Let a shape $\mathsf{S}$ of the forward curve with $k$ local extrema be given. Choosing a suitable Descartes-subsystem $\cD'$ of $\cD$ with $k+1$ elements, we can apply Theorem~\ref{thm:extremal} and find a D-polynomial $f$ in $\cD'$, such that $f$ has a sign sequence with $k$ sign changes, which corresponds to the shape $\mathsf{S}$. Padding the list of coefficients with zeroes, we can write $f$ as a D-polynomial in the full system $\cD$, i.e. as
\begin{equation*}
f(x) = a_{2\lambda_2} \varphi_{2\lambda_2}(x) + a_{\lambda_1 + \lambda_2} \varphi_{\lambda_1 + \lambda_2}(x) + a_{\lambda_2} \varphi_{\lambda_2}(x) +  a_{2 \lambda_1} \varphi_{2\lambda_1}(x) + a_1 \varphi_{\lambda_1}(x), 
\end{equation*}
where we have labeled the coefficients $a$ consistently with the basis functions of $\cD$. Comparing coefficients with \eqref{eq:l_poly}, we can conclude that the shape $\mathsf{S}$ is attainable in the Vasicek-model, if we can show that the system of equations
\begin{subequations}\label{eq:key}
\begin{align}
\frac{\sigma_1^2}{\lambda_1} &= a_{2\lambda_1} \label{eq:key_sigma1}\\
\frac{\sigma_2^2}{\lambda_2} &= a_{2\lambda_2} \label{eq:key_sigma2}\\
\rho (\lambda_1 + \lambda_2) \frac{\sigma_1\sigma_2}{\lambda_1\lambda_2} &= a_{\lambda_1 + \lambda_2} \label{eq:rho}\\
 \lambda_1 \left(\theta_1 -z_1\right) - \frac{\sigma_1^2}{\lambda_1} - \rho \lambda_1 \frac{\sigma_1\sigma_2}{\lambda_1\lambda_2}  &= a_{\lambda_1}\label{eq:key_z1}\\
 \lambda_2 \left(\theta_2 -z_2\right) - \frac{\sigma_2^2}{\lambda_2} - \rho \lambda_2 \frac{\sigma_1\sigma_2}{\lambda_1\lambda_2}  &= a_{\lambda_2}\label{eq:key_z2}
\end{align}
\end{subequations}
has a solution $(\sigma_1, \sigma_2, \rho, z_1, z_2) \in [0,\infty)^2 \times [-1,1] \times \RR^2$. The argument for yield curves is analogous, using the appropriate Descartes system $\cE$ from Lemma~\ref{lem:Descartes_g}.\\
Having reduced the attainability problem to the equation system \eqref{eq:key}, we need to discuss its solvability: Clearly, whenever \eqref{eq:key_sigma1} -- \eqref{eq:rho} can be solved for $(\sigma_1, \sigma_2, \rho)$, then also \eqref{eq:key_z1} and \eqref{eq:key_z2} can be solved for $(z_1, z_2)$. Moreover, the solvability of \eqref{eq:key_sigma1} and \eqref{eq:key_sigma2} for $(\sigma_1, \sigma_2)$ only depends on the signs of $a_{2\lambda_1}$ and $a_{2\lambda_2}$. It is therefore only \eqref{eq:rho} for which solvability is nontrivial, due to the restriction $\rho \in [-1,1]$. These elementary observations are summarized in the following Lemma:
\begin{lem}\label{lem:key} Consider the system of equations given in \eqref{eq:key}
\begin{enumerate}[(a)]
\item If $a_{2\lambda_1} < 0$ or $a_{2\lambda_2} < 0$, then \eqref{eq:key} has no solution.
\item If $a_{2\lambda_1} = 0$ and $a_{2\lambda_2} \ge 0$, or if $a_{2\lambda_1} \ge 0$ and $a_{2\lambda_2} = 0$ then \eqref{eq:key} has a solution. In this solution $\sigma_1 = \rho = 0$ or $\sigma_2 = \rho = 0$ or both.
\item If $a_{2\lambda_1} > 0$ and $a_{2\lambda_2} > 0$, then \eqref{eq:key} has a solution if and only if 
\begin{equation}
\rho := \frac{\sqrt{\lambda_1 \lambda_2}}{\lambda_1  + \lambda_2 } \frac{a_{\lambda_1 + \lambda_2}}{\sqrt{a_{2\lambda_1}a_{2\lambda_1}}} \quad \text{is in $[-1,1]$.}
\end{equation}
\end{enumerate}
\end{lem}

To complete the proof of Theorem~\ref{thm:main} we apply the strategy outlined above on a case-by-case basis to the different shapes:

\begin{proof}[Proof of Theorem~\ref{thm:main} -- sufficiency]
We partition the proof according to the number $k$ of local extrema of the term structure curve; later we also need to distinguish between the cases (a), (b) and (c) given in Theorem~\ref{thm:main}.
\begin{enumerate}[(i)]
\item For $k=0$ we use the system $\cD_1 = (\varphi_{\lambda_1})$. We set $a^\pm_{\lambda_1} = \pm 1$ and all other coefficients to zero. This yields the D-polynomials $\varphi_\pm(x) = \pm \varphi_{\lambda_1}(x) = \pm e^{-\lambda_1 x}$ with sign sequences $\sise{\pp}$ and $\sise{\mm}$. Setting $z_2 = \sigma_1 = \sigma_2 = \rho = 0$ the system \eqref{eq:key} can be solved for $z_1$ in both cases. We conclude that the shapes \normal{} and \inverse{} are attainable.
\item For $k=1$ we use the system $\cD_2 = (\varphi_{\lambda_2}, \varphi_{\lambda_1})$. By Theorem~\ref{thm:extremal} we can find two extremal D-polynomials $\varphi_+, \varphi_-$ with coefficients $(a^\pm_{\lambda_2}, a^\pm_{\lambda_1})$ and sign sequences $\sise{\pp \mm}$ and $\sise{\mm \pp}$. Setting $\sigma_1 = \sigma_2 = \rho = 0$ the system \eqref{eq:key} can be solved for $(z_2, z_1)$ in both cases. We conclude that the shapes \dip{} and \humped{} are attainable.
\item For $k=2$ we use the system $\cD_3 = (\varphi_{2 \lambda_2}, \varphi_{\lambda_2}, \varphi_{\lambda_1})$. By Theorem~\ref{thm:extremal} we can find two extremal D-polynomials $\varphi_+, \varphi_-$ with coefficients $(a^\pm_{2 \lambda_2}, a^\pm_{\lambda_2}, a^\pm_{\lambda_1})$ and sign sequences $\sise{\pp \mm \pp}$ and $\sise{\mm \pp \mm}$. Setting $\sigma_1 = \rho = 0$ the system \eqref{eq:key} can be solved for $(\sigma_2, z_2, z_1)$ in the case of $\varphi_+$. In the case of $\varphi_-$ the system cannot be solved, because $a^-_{2 \lambda_2} < 0$. We conclude that the shape \hudi{} is attainable. 
\setcounter{mycounter}{\value{enumi}}
\end{enumerate}
At this point we have already covered all attainable shapes in the scale-proximal case with $\rho \ge 0$, i.e., part (b) of the theorem. Next we complete part (a), i.e., the scale-separated case:
\begin{enumerate}[(i)]
\setcounter{enumi}{\value{mycounter}}
\item For $k=2$ we can alternatively use the system $\cD_{3,sep} = (\varphi_{\lambda_2}, \varphi_{2 \lambda_1}, \varphi_{\lambda_1})$, which is a subsystem of $\cD_\text{sep}$.\footnote{But not a Descartes subsystem of $\cD_\text{prox}$!} By Theorem~\ref{thm:extremal} we can find two extremal D-polynomials $\varphi_+, \varphi_-$ with coefficients $(a^\pm_{\lambda_2}, a^\pm_{2 \lambda_1}, a^\pm_{\lambda_1})$ and sign sequences $\sise{\pp \mm \pp}$ and $\sise{\mm \pp \mm}$. Setting $\sigma_2 = \rho = 0$ the system \eqref{eq:key} can be solved for $(z_2, \sigma_1, z_1)$ in the case of $\varphi_-$. In the case of $\varphi_+$ the system cannot be solved, because $a^-_{2 \lambda_1} < 0$. We conclude that the shape \dihu{} is attainable. 
\item For $k=3$, we use the system $\cD_\text{4,sep} =  (\varphi_{2 \lambda_2} \varphi_{\lambda_2}, \varphi_{2 \lambda_1}, \varphi_{\lambda_1})$, which is a subsystem of $\cD_\text{sep}$.  By Theorem~\ref{thm:extremal} we can find two extremal D-polynomials $\varphi_+, \varphi_-$ with coefficients $(a^\pm_{2 \lambda_2}, a^\pm_{\lambda_2}, a^\pm_{2 \lambda_1}, a^\pm_{\lambda_1})$ and sign sequences $\sise{\pp \mm \pp \mm}$ and $\sise{\mm \pp \mm \pp}$. Setting $\rho = 0$ the system \eqref{eq:key} can be solved for $(\sigma_2, z_2, \sigma_1, z_1)$ in the case of $\varphi_+$. In the case of $\varphi_-$ the system cannot be solved, because $a^-_{2 \lambda_1} < 0$ and $a^-_{2 \lambda_2} < 0$. We conclude that \HDH{} is attainable. 
\setcounter{mycounter}{\value{enumi}}
\end{enumerate}
At this point we have also covered all attainable shapes in the scale-separated case (with arbitrary $\rho$) and thus part (a) is complete. The most difficult case is part (c), i.e., the scale-proximal case with $\rho < 0$. Here, Theorem~\ref{thm:extremal} is not sufficient to find suitable D-polynomials $\varphi_\pm$ and we have to use the more specialized result Lemma~\ref{lem:special} instead.

\begin{enumerate}[(i)]
\setcounter{enumi}{\value{mycounter}}
\item For $k=3$ we use the system $\cD_{4,prox} = (\varphi_{2 \lambda_2}, \varphi_{\lambda_1 + \lambda_2}, \varphi_{2 \lambda_1}, \varphi_{\lambda_2})$, which is a subsystem of $\cD_\text{prox}$. By Lemma~\ref{lem:special} we can find two sets of real numbers $0 < r_1^+ < r_2^+ < r_3^+$ and $0 = r_1^0 < r_2^0 < r_3^0$ as well as D-polynomials $\varphi_+$ and $\varphi_0$ with the following properties: 
\begin{itemize}
\item The zeroes of $\varphi_+$ and $\varphi_0$ are located exactly at the points $r_1^+, r_2^+, r_3^+$ and $r_1^0 = 0, r_2^0, r_3^0$;
\item the sign sequence of $\varphi_+ $ is $\sise{\pp \mm \pp \mm}$ and the sign sequence of $\varphi_0$ is $\sise{\mm \pp \mm}$;
\item the coefficients of both $\varphi_+$ and $\varphi_0$ have sign sequence $\sise{\pp \mm \pp \mm}$.
\end{itemize}
 Moreover, the coefficients (of both $\varphi_+$ and $\varphi_0$) satisfy
\[\left| \frac{a_{\lambda_1 + \lambda_2}}{\sqrt{a_{2 \lambda_1}a_{2 \lambda_2}}} \right| < 2;\]
see \eqref{eq:coef_ineq}. Thus, applying the geometric-arithmetic-mean inequality, we obtain
\begin{equation}\label{eq:rho_solvable}
|\rho| \le \frac{\sqrt{\lambda_1 \lambda_2}}{\lambda_1 + \lambda_2} \left| \frac{a_{\lambda_1 + \lambda_2}}{\sqrt{a_{2 \lambda_1}a_{2 \lambda_2}}} \right| < 1.
\end{equation}
By Lemma~\ref{lem:key}, this implies that the system of equations \eqref{eq:key} is solvable. We conclude that the shapes \HDH{} and \dihu{} are attainable. 
\item For $k=4$ we use the full system $\cD_\text{prox}$. As in the previous case, we can apply Lemma~\ref{lem:special} to find two D-polynomials $\varphi_+$ and $\varphi_0$ with prescribed zeroes and with sign sequences $\sise{\pp \mm \pp \mm \pp}$ and $\sise{\mm \pp \mm \pp}$ respectively. The first zero of $\varphi_0$ is located at the boundary point $r_1^0 = 0$. Moreover, the coefficients of both $\varphi_+$ and $\varphi_0$ have sign sequence $\sise{\pp \mm \pp \mm \pp}$ and inequality \eqref{eq:rho_solvable} holds. Thus, Lemma~\ref{lem:key} implies that the system of equations \eqref{eq:key} is solvable and we conclude that the shapes \HDHD{} and \DHD{} are attainable.
\end{enumerate}
Having completed part (c), also the last case of Theorem~\ref{thm:main} is shown. The scale-critical case can be treated like the scale-proximal case if $\rho \ge 0$, and like the scale-separated case if $\rho < 0$. 
\end{proof}

\section{Additional Results}\label{sec:additional}

\subsection{State-contingent analysis of term structure shapes}\label{sec:state}
Using the same general ideas as in the previous section, the analysis of term structure shapes can also be carried out \emph{contingent on the state vector $(z_1, z_2) \in \RR^2$.} In other words, the state space $\RR^2$ can be partitioned into regions in which only a few -- often only a single -- term structure shape can occur; see Figure~\ref{fig:1}.

We use four equations to constrain the shape of the term structure. Two are derived from Theorem~\ref{thm:sise_m} and constrain the overall shape of the term structure; the other two are derived from the initial and the terminal sign of the derivative of the yield/forward curve. All equations are linear, i.e., the partitions of the state space can be described as intersections of half-spaces. 
It will be convenient to reparameterize as
\[y_i = z_i - \theta_1\qquad a_i = \sigma_i / \lambda_i.\]
Under this change of variables, the quantities $w_1, w_2$ from Lemma~\ref{lem:curve_diff} become
\begin{align*}
w_1(y_1,y_2) &= - y_1 - (a_1^2 + \rho a_1 a_2)\\
w_2(y_1,y_2) &= - y_2 - (a_2^2 + \rho a_1 a_2).
\end{align*}
These linear functions determine the half spaces
\[W_i^- = \set{(y_1,y_2) \in \RR^2: w_i(y_1,y_2) \le 0}, \qquad i \in \set{1,2},\]
and, with reversed inequalities, $W_{i}^+$. On each of the possible combinations $(W_1^\pm, W_2^\pm)$, the pair $(w_1,w_2)$ has a different combination of signs, and the resulting restrictions on the term structure shape can be read from Theorem~\ref{thm:sise_m}. 

The second pair of inequalities is obtained from the initial and the terminal sign of the term structure curve's derivative. Both quantities can be derived from Lemma~\ref{lem:curve_diff}. The initial (`\underline{s}tarting') sign of the derivative is the same for yield and forward curve and is equal to the sign of 
\[s(y_1,y_2) = u_1 + u_2 + c + w_1(y_1,y_2) + w_2(y_1,y_2) = - \lambda_1 y_1 - \lambda_2 y_2,\]
which gives a linear function defining the half spaces $S^\pm$. 
The terminal sign of the derivative of the yield and forward curve is different; for the forward curve it is equal to the sign of $\lim_{x \to \infty} e^{\lambda_1 x} \partial_x f(x) = w_1(y_1,y_2)$.\footnote{The fact that the terminal sign of $\partial_x f$ is equal to $w_1$ was the key observation that led to Cor.~\ref{cor:sise_l}.} For the yield curve, it is equal to the sign of 
\[t_{\text{yield}}(y_1,y_2) = \lim_{x \to \infty} x^2 \partial_x Y(x) = -y_1 - y_2 - \frac{1}{2}\left(a_1^2 + a_2^2 + 2\rho a_1a_2\right), \]
defining the half spaces $T_\text{yield}^\pm$. 

Overall, there are up to $2^4 = 16$ combinations of the half spaces $(W_1^\pm, W_2^\pm, S^\pm, T_\text{yield}^\pm)$ leading to different restrictions on the yield curve, and $2^3 = 8$ combinations for the forward curve. The actual configuration of half spaces, and hence the number and shape of intersections, depends on the model regime, i.e., scale-proximity vs. scale-separation and on the sign of the correlation parameter $\rho$. A full analysis of all cases is beyond the scope of the paper, but we add some details to the scale-proximal, positively correlated case, which is shown in Figure~\ref{fig:1}. We consider two exemplary cases for the analysis of the forward curve:
\begin{enumerate}[(a)]
\item On the intersection $W_1^- \cap W_2^- \cap S^+$ the pair $(w_1, w_2)$ has two negative signs. Hence, via Theorem~\ref{thm:sise_m}, the sign sequence of $\partial_x f$ is a subsequence of $\sise{\pp \mm}$. The terminal sign of $\partial_x f$ is equal to the sign of $w_1$, hence also negative. This leaves $\sise{\pp \mm}$ and $\sise{\mm}$ as possibles sign sequences of $\partial_x f$, corresponding to the shapes \humped{} and \inverse{}. The final half space $S^+$ restricts the inital sign to \pp{} and selects the unique remaining possibility of a \humped{} forward curve.
\item On the intersection $W_1^- \cap W_2^+ \cap S^+$ the pair $(w_1, w_2)$ has sign sequence $\sise{\mm \pp}$. From Theorem~\ref{thm:sise_m} we obtain the \emph{same} restriction $\sseq(\partial_x f) \subset \sise{\pp \mm}$ as in case (a). The restrictions on initial and terminal sign are also the same, such that the forward curve must also be \humped{} in this case. 
\end{enumerate}
All other intersections of half spaces can be analyzed in the same way. For the yield curve, the analysis in case (a) must be adapted as follows:
\begin{enumerate}[(a')]
\item From Theorem~\ref{thm:sise_m} the restriction $\sseq(\partial_x Y) \subset \sise{\pp \mm}$ is obtained, and the initial sign of the yield curve must be positive on $S^+$. This leaves the possibilities $\sise{\pp \mm}$ and $\sise{\mm}$ for the yield curve's derivative, corresponding to a \humped{} or \inverse{} curve. Contrary to the forward curve there is no restriction to the terminal sign of the curve. Intersecting with the half space $T_\text{yield}^+$ selects the \humped{} curve, intersecting with $T_\text{yield}^-$ selects the \inverse{} curve.
\end{enumerate}
In the scale-proximal, positively correlated case, this type of analysis results in a unique curve shape for all intersections, except for $W_1^+ \cap W_2^- \cap S^+ (\cap\;T_\text{yield}^+)$. Only for this part of the state space, the method is not able to differentiate between a \normal{} and a \hudi{} shape, see also Figure~\ref{fig:1}.

\subsection{Strict, strong, and $\Sigma$-attainability}
The results of Theorem~\ref{thm:main} can be sharpened by introducing the notions of strict, strong, and $\Sigma$-attainability. 
\begin{defn}[Strict and strong attainability]\label{def:strong_attainable}
\hfill
\begin{enumerate}[(a)]
\item A shape $\mathsf{S}$ of the forward curve is called \textbf{strictly attainable}, if we can find a parameter vector $p \in \bm{P}$, such that  $x \mapsto f(x; Z_t, p)$ attains shape $\mathsf{S}$ with strictly positive probability for all $t > 0$.
\item A shape $\mathsf{S}$ of the forward curve with $k$ local extrema is called  \textbf{strongly attainable}, if for any $0 < r_1 < \dotsm < r_k$, we can find a parameter vector $p \in \bm{P}$ and a state vector $z \in \RR^2$, such that $x \mapsto f(x; z, p)$ has shape $\mathsf{S}$, with its local extrema located at $r_1, \dotsc, r_k$. In other words, strong attainability means that the locations (but not the amplitude!) of the extrema of the forward curve can be chosen arbitrarily.
\end{enumerate}
The same terminology is applied to the yield curve $x \mapsto Y(x;z,p)$.
\end{defn}
\begin{rem}
We remark that in (a) it makes no difference whether probabilities under the risk-neutral measure $\QQ$ or probabilities under the statistical measure $\PP$ are considered, as $\QQ$ and $\PP$ are equivalent. It also makes no difference whether `all $t > 0$' or `some $t > 0$' are considered, as in the Vasicek model also the laws of $Z_t$ and $Z_{t'}$ are equivalent for any $t,t' > 0$. 
\end{rem}

Another strengthening of Theorem~\ref{thm:main} can be obtained by varying not all parameters in $\bm{P}$, but only a subset of them. To formulate these results, we write $\bm{P}'$ for $\bm{P}$ with the volatility parameters $(\sigma_1, \sigma_2, \rho)$ removed, and introduce the parameter space of covariance matrices
\[\bm{\Sigma} := \set{\Sigma  = \begin{pmatrix} \sigma_1^2 & \rho \sigma_1 \sigma_2 \\ \rho \sigma_1 \sigma_2 & \sigma_2^2 \end{pmatrix}, \sigma_1, \sigma_2 \in [0,\infty), \rho \in [-1,1]}.\]
Additional restrictions on $\bm{\Sigma}$ are denoted by $\bm{\Sigma}_{\rho <0}$, $\bm{\Sigma}_{\rho = 0}$, etc. We can now introduce the notion of $\Sigma$-attainability:
\begin{defn}
A shape $\mathsf{S}$ of the forward curve is called \textbf{$\Sigma$-attainable}, if for any parameter vector $p' \in \bm{P}'$, we can find a covariance matrix $\Sigma \in \bm{\Sigma}$ and a state vector $z \in \RR^2$, such that  $x \mapsto f(x; z, (p',\Sigma))$ has shape $\mathsf{S}$.\\
The same terminology is applied to the yield curve $x \mapsto Y(x;z, (p',\Sigma))$.
\end{defn}
Combining with Definition~\ref{def:strong_attainable} we also obtain the notions of \textbf{strict and strong $\Sigma$-attainability}.\\

The first corollary to Theorem~\ref{thm:main} that we present, concerns the $\Sigma-$ and the \emph{strong} attainability of term structure curves.

\begin{cor}\label{cor:strong}
\begin{enumerate}[(i)]
\item In all cases of Theorem~\ref{thm:main}, the given shapes are $\Sigma$-attainable, even when restricted to regular covariance matrices only.
\item In cases (a) and (b) the shapes are strongly attainable, and even strongly $\Sigma$- and $\Sigma_{\rho=0}$-attainable. 
\item In case (c), all shapes except possibly \dihu{},  \HDH{}, \DHD{}, and \HDHD{} are strongly attainable, and even $\Sigma$- and $\Sigma_{\rho<0}$-attainable.
\end{enumerate}
\end{cor}

The second corollary concerns the \emph{strict} attainability.

\begin{cor}\label{cor:strict}
\begin{enumerate}[(i)]
\item The shapes of Theorem~\ref{thm:main}(a) are strictly attainable, and even strictly $\Sigma_{\rho > 0}$-, $\Sigma_{\rho = 0}$-, and $\Sigma_{\rho < 0}$-attainable.
\item The shapes of Theorem~\ref{thm:main}(b) are strictly attainable, and even strictly $\Sigma_{\rho > 0}$- and $\Sigma_{\rho = 0}$-attainable
\item The shapes of Theorem~\ref{thm:main}(c), with possible exception of $\dihu$ and $\DHD{}$, are strictly attainable, and even strictly $\Sigma_{\rho < 0}$-attainable. 
\end{enumerate}
\end{cor}

\begin{rem}
We give a heuristic argument supporting Cor.~\ref{cor:strong}, which aims to explain why the variation of the (co-)variance parameters and the state vector is sufficient to attain the listed shapes: 
\begin{itemize}
\item For strong attainability of \HDH{}, the most complex shape in case (a) of Thm.~\ref{thm:main}, four degrees of freedom are needed: Three for the local extrema and an additional degree of freedom to select between \HDH{} and \DHD{}. The parameter space $\bm{\Sigma}_{\rho = 0}$ has two degrees of freedom and the state space $\RR^2$ also has two, matching the required four degrees. 
\item In case (c) of Thm.~\ref{thm:main} the most complex shape, \HDHD{}, needs five degrees of freedom. The parameter space $\bm{\Sigma}_{\rho < 0}$ provides three of them and the state space $\RR^2$ provides two. 
\item In case (b) the positive correlation of the two factor processes together with the proximity of the mean-reversion scales leads to so much positive reinforcement that not all five degrees of freedom can be utilized. 
\end{itemize}
\end{rem}

We now explain how the stronger conclusions of Corollary~\ref{cor:strong} and \ref{cor:strict} can be obtained from the proof of Theorem~\ref{thm:main} that was given in Sec.~\ref{sec:attainable}. First, observe that in all steps (i) - (vii) of the proof, we have shown that the system of equations \eqref{eq:key} could be solved by choosing suitable covariance parameters $(\sigma_1, \sigma_2, \rho)$ and state vectors $(z_1, z_2)$ and that it was not necessary to modify any of the remaining parameters in $\bm{P}'$. This shows that attainability can be strengthened to \emph{$\Sigma$-attainability} in all cases.\\
Next, observe that that in steps (i) - (v) of the proof we have used Theorem~\ref{thm:extremal} to find a D-polynomial $\varphi_+$ or $\varphi_-$, which, after solving \eqref{eq:key}, equates to $\partial_x f$, the derivative of the forward curve. Theorem~\ref{thm:extremal} allows us to predetermine all zeroes $r_1 < \dotsm < r_k$ of $\varphi_\pm$, and hence the locations of the extrema of the forward curve. The same is true for $\partial_x Y$, the derivative of the yield curve. This shows that in cases (i) -(v) we obtain \emph{strong} $\Sigma$-attainability. In addition, note that it was sufficient to choose $\rho = 0$ in all cases (i) - (v). Thus, we even get \emph{strong $\Sigma_{\rho = 0}$-attainability}. This completes the arguments needed for Cor.~\ref{cor:strong}.\\

The contents of Cor~\ref{cor:strict} follow from a perturbation argument. Consider for instance case (iii) in the proof of Thm.~\ref{thm:main}: There, we have shown that we can find parameters $\sigma_1 = \rho = 0$, $\sigma_2 > 0$ and a state vector $(z_1,z_2) \in \RR$, which produces the sign sequence $\sise{\pp \mm \pp}$ corresponding to shape $\hudi{}$. Suppose that a perturbation 
\[\sigma_1^\epsilon = \epsilon, \quad \rho^\epsilon = \pm \epsilon \quad \text{and} \quad z_1^\epsilon = z_1 \pm \epsilon, \quad z_2^\epsilon = z_2 \pm \epsilon\] 
with $\epsilon$ in some small set $[0,\delta)$ still produces the same sign sequence $\sise{\pp \mm \pp}$ and shape $\hudi{}$. Then, we may conclude
\begin{itemize}
\item that the shape $\hudi{}$ is strictly $\Sigma$-attainable, as $(Z_t^1, Z_t^2)$ visits any small neighborhood of $(z_1,z_2)$ with strictly positive probability; 
\item that $\hudi{}$ is also $\Sigma_{\rho > 0}$- and $\Sigma_{\rho < 0}$-attainable, as we have relaxed the condition $\rho = 0$ to $\rho^\epsilon = \pm \epsilon$; and
\item that it is sufficient to consider regular matrices $\Sigma$, as we have relaxed the condition $\sigma_1 = 0$ to $\sigma_1^\epsilon = \epsilon$.
\end{itemize}
The necessary perturbation Lemma is given below. Applying the same argument to each of the cases (i) - (v) in the proof yields part (a) and (b) of Cor.~\ref{cor:strict}. For cases (vi) and (vii) note that the Lemma can only be applied to the D-polynomial $\varphi_+$, but not to $\varphi_0$, which has a zero at the boundary of $[0,\infty)$ and is not an extremal D-polynomial. This yields part (c) of Cor.~\ref{cor:strict}.

\begin{lem}[Perturbation Lemma]
Let $\phi = \sum_{i=1}^n a_i \phi_i$ be a non-vanishing D-polynomial in a Descartes system $\cD = (\phi_1, \dotsc, \phi_n)$ on a subinterval $X \subset \RR$ which satisfies
\[\sseq\left(\sum_{i=1}^n a_i \phi_i\right) \simeq \sise{a_1 \dotsc a_n}\]
and has no zeroes on the boundary of $X$. Then there exist $(b_i)_{i=1\dotsc n} \in \set{-1,+1}$ and $\delta > 0$, such that
\[\sseq\left(\sum_{i=1}^n a_i^\epsilon \phi_i\right) \simeq \sseq\left(\sum_{i=1}^n a_i \phi_i\right)\]
for all $\epsilon \in [0,\delta)$ and with $a_i^\epsilon = a_i + \epsilon b_i$.
\end{lem}
\begin{proof}
First, we show that the sequence $(b_i)$ can be chosen such that 
\[\sise{a_1^\epsilon \dotsc a_n^\epsilon} \simeq \sise{a_1 \dotsc a_n}.\] 
To this end define $b_1, \dotsc, b_n$ as follows:
\begin{align*}
a_i  > 0 \quad &\Longrightarrow &\quad &b_i := +1\\
a_i < 0  \quad &\Longrightarrow &\quad &b_i := -1\\
a_i = 0 \quad &\Longrightarrow &\quad &b_i := \begin{cases}+1 \quad &\text{if the block of zeroes containing $a_i$ borders}\\&\;\text{on at least one $a_j > 0$,} \\-1 \quad &\text{else.}\end{cases}
\end{align*}
It is easy to see that the number and direction of strong sign changes in $(a_1^\epsilon, \dotsc, a_n^\epsilon)$ is the same as in $(a_1, \dotsc, a_n)$ for all $\epsilon \ge 0$, i.e., we have
\begin{equation*}
\sise{a_1^\epsilon \dotsc a_n^\epsilon} \simeq \sise{a_1 \dotsc a_n}, \quad \forall\,\epsilon \ge 0.
\end{equation*}
Set $\phi^\epsilon = \sum_{i=1}^n a_i^\epsilon \phi_i$. Then by Theorem~\ref{thm:vardim}
\begin{equation}\label{eq:sise_sub}
\sseq(\phi^\epsilon)  \subseteq \sise{a_1^\epsilon \dotsc a_n^\epsilon} \simeq \sise{a_1 \dotsc a_n} \simeq \sseq(\phi),
\end{equation}
for all $\epsilon \ge 0$, and we have shown that $\phi^\epsilon$ cannot have \emph{more} sign changes than $\phi$. It remains to show that equivalence holds for small enough $\epsilon$. Let $k$ be number of strong sign changes of $\phi$. Clearly, we can find $r_0, \dotsc, r_k$ such that the sequence $\phi(r_i)_{i=0, \dotsc, k}$ is of alternating signs. Each interval $(r_i, r_{i+1})$ must contain exactly one zero of $\phi$. Set
\[\delta := \frac{\min_{i=0, \dotsc, k} |\phi(r_i)|}{\sum_{j=1}^n \max_{i=0, \dotsc, k} |\phi_j(r_i)|}.\]
Then, $\delta > 0$ and for all $\epsilon \in [0,\delta)$
\begin{align*}
\left|1 - \frac{\phi^\epsilon(r_i)}{\phi(r_i)} \right| &= \left|\frac{\phi(r_i) - \phi^\epsilon(r_i)}{\phi(r_i)} \right| \le \frac{1}{|\phi(r_i)|} \left| \sum_{j=1}^n \epsilon b_j \phi_j(r_i) \right| \le \\ &\le \epsilon \frac{\sum_{j=1}^n |\phi_j(r_i)|}{|\phi(r_i)|} < 1.
\end{align*}
This shows that the sequence $\phi^\epsilon(r_i)_{i=0, \dotsc, k}$ has the same alternating signs as $\phi(r_i)_{i=0, \dotsc, k}$ and hence that $\phi^\epsilon$ has at least the same number of zeroes as $\phi$, for all $\epsilon \in [0,\delta)$. Together with \eqref{eq:sise_sub}, this completes the proof.
\end{proof}

\appendix

\section{Auxilliary results on Descartes systems}

Let a family $(\phi_1, \dotsc, \phi_k)$ of functions on $X \subseteq \RR$ be given. We set $\bm{x} = (x_1, \dotsc, x_k) \in X^k$ and 
\[\Delta_k(X) := \set{\bm{x} \in X^k: x_1 < \dotsc < x_k}.\] 
From \cite{borwein1995polynomials} we adopt the compact notation
\begin{equation}\label{eq:det}
D\begin{pmatrix}\phi_1, \dotsc ,\phi_k\\x_1, \dotsc, x_k\end{pmatrix} := \det \begin{pmatrix} \phi_{1}(x_1) & \phi_{2}(x_1) & \dotsc &  \phi_{k}(x_1)\\ \vdots & \vdots && \vdots\\  \phi_{1}(x_k) &  \phi_{2}(x_k) & \dotsc &  \phi_{k}(x_k)\end{pmatrix}.
\end{equation}
An important special case is the Vandermonde determinant, which for any real $(\gamma_i)_{i=1, \dotsc, k}$ evaluates as
\begin{equation}\label{eq:Vandermonde}
D\begin{pmatrix}1, x ,x^2, \dotsc, x^{k-1}\\ \gamma_1, \dotsc, \gamma_k\end{pmatrix} = \prod_{j=1}^{k-1} (\gamma_j - \gamma_{j-1}), 
\end{equation}
see e.g. \cite[Ch.~22.4]{hogben2013handbook}. For sufficiently differentiable functions $\phi_1, \dotsc, \phi_k$, we also introduce the Wronskian determinant (or simply Wronskian)
\begin{equation}\label{eq:wronskian}
W(\phi_1, \dotsc, \phi_k)(x) = \det \begin{pmatrix} 1 & \phi_1(x) & \phi_1'(x) & \dotsm & \phi_1^{(k)}(x) \\ 1 & \phi_2(x) & \phi_2'(x) & \dotsm & \phi_2^{(k)}(x) \\ \vdots & \vdots & \vdots & & \vdots\\ 1 & \phi_k(x) & \phi_k'(x) & \dotsm & \phi_k^{(k)}(x) \\\end{pmatrix}.
\end{equation}
In \cite[Ch.~2, \S 2]{karlin1968total} relations between the two determinants in \eqref{eq:det} and \eqref{eq:wronskian} as well as intermediate notions of `derivated determinants' are discussed. 

\subsection{D-polynomials with prescribed zeroes}\label{app:interpolation}

\begin{proof}[Proof of Theorem~\ref{thm:extremal}]
Let a Descartes system $\cD = (\phi_1, \dotsc, \phi_n)$ on $X$ and a set of prescribed zeroes $\bm{r} =(r_1, \dotsc, r_{n-1}) \in \Delta_{n-1}(X)$ be given. We show that the D-polynomial  
\begin{equation}\label{eq:interpol}
\phi(x;\bm{r}) = D\begin{pmatrix} \phi_1, &\phi_2, &\dots, &\phi_n\\ x, &r_1, &\dotsc ,&r_{n-1} \end{pmatrix} 
\end{equation}
is the desired interpolation polynomial of Theorem~\ref{thm:extremal}. First, observe that the determinant vanishes whenever $x = r_i$ for any $i = 1, \dotsc, n-1$, and hence $\phi(x,\bm{r})$ possesses a zero at each $r_i$, which shows (a). Second, as $\cD$ is a Descartes system, the determinant must be non-zero at all other points in $X$. The point $x$ crossing an interior zero $r_i$ changes the order of two columns in the determinant and hence flips the sign of $\phi(x;\bm{r})$, which shows (b). Claim (c) now follows from Theorem~\ref{thm:vardim} -- because $\phi(x;\bm{r})$ has $n-1$ sign changes, equivalence must hold in \eqref{eq:vardim}. 
\end{proof}
To prepare for additional results, we remark that the coefficients $a_1, \dotsc, a_n$ of the interpolation D-polynomial $\phi(x;\bm{r})$ can be determined directly from \eqref{eq:interpol}. Expanding the determinant in the first column yields 
\[\phi(x,\bm{r}) = \sum_{i=1}^n a_i(\bm{r}) \phi_i(x),\]
where 
\begin{equation}
\label{eq:coefficients}
a_i(\bm{r}) = (-1)^{1+i} D\begin{pmatrix}\phi_1, \dotsc, \phi_{i-1}, \; \phi_{i+1}, \dotsc, \phi_n \\ r_1, \dotsc, r_{n-1}\end{pmatrix}.
\end{equation}
Because $(\phi_1, \dotsc, \phi_n)$ is a Descartes system, the determinant on the right hand side is strictly positive. This shows that the coefficients of $\phi(x,\bm{r})$ must have alternating signs, starting with $\pp$.

\subsection{The Descartes property of $\cE$} \label{app:Descartes_E}

\begin{proof}[Proof of Lemma~\ref{lem:Descartes_g}]
To show that $\cE_\text{sep}, \cE_\text{prox}$ and $\cE_\text{crit}$ are Descartes systems on $[0,\infty)$, it is sufficient to show that
\[D\begin{pmatrix}g_{\alpha_k}, \dotsc, g_{\alpha_1}\\ x_1, \dotsc, x_k\end{pmatrix} > 0\]
for any $\alpha_k > \dotsc > \alpha_1 \ge 0$ and $\bm{x} = (x_1, \dotsc, x_k) \in \Delta_{k-1}[0,\infty)$.  Our starting point is the representation \eqref{eq:g_rep} of $g_\alpha$ as an integral of $\varphi_\alpha(x) = e^{-\alpha x}$ with respect to the totally positive kernel
\[K(x,y) = \frac{y}{x^2}\Ind{x \le y}.\]
From \cite[Ch.~3, Eq.~(1.11)ff]{karlin1968total} and with $K_i := K(x,y_i)$ we obtain that
\[D\begin{pmatrix}K_1, \dotsc, K_k\\ x_1, \dotsc, x_k\end{pmatrix}  = \frac{y_1 \dotsm y_{k}}{x_1^2 \dotsm x_k^2} \Ind{0 \le y_1 \le x_1 \le y_2 \le x_2 \dotsm \le x_k}\]
for any $\bm{x}, \bm{y} \in \Delta_k(0,\infty)$. Combining this with the composition formula \cite[Ch.~3, Eq.~(1.2)]{karlin1968total} we obtain
\begin{multline}
D\begin{pmatrix}g_{\alpha_k}, \dotsc, g_{\alpha_1}\\ x_1, \dotsc, x_k\end{pmatrix} = \\ = \int_0^{x_1} \int_{x_1}^{x_2} \dotsm \int_{x_{k-1}}^{x_k} \frac{y_1  \dotsm y_{k}}{x_1^2 \dotsm x_k^2} D\begin{pmatrix}\varphi_{\alpha_k}, \dotsc, \varphi_{\alpha_1}\\ x_1, \dotsc, x_k\end{pmatrix} \,dy_1 dy_2 \dotsm dy_k.
\end{multline}
Because $(\varphi_{\alpha_k}, \dotsc, \varphi_{\alpha_1})$ is a Descartes system, the integrand is strictly positive. Moreover, the domain of integration has strictly positive measure. We conclude that the left hand side is strictly positive for any $\bm{x} = (x_1, \dotsc, x_k) \in \Delta_k(0,\infty)$, and hence that $\cE$ is a Descartes system on $(0,\infty)$. It remains to extend this property to the left-closed interval $[0,\infty)$. By \cite[Ch.~2, Thm.~2.3]{karlin1968total} it is sufficient to show the Wronskian $W(g_{\alpha_k}, \dotsc, g_{\alpha_1})(0)$ is strictly positive for any $k$. We first calculate the Taylor expansion 
\[g_\alpha(x) = \frac{1}{x^2} \int_0^x y e^{-\alpha y} dy = \sum_{k=0}^\infty \frac{(-\alpha)^k}{(k +2)} \frac{x^k}{k!},\]
which follows from the Taylor expansion of the exponential function. We conclude that the $k$-th derivative of $g_\alpha$ at zero is given by
\begin{equation}\label{eq:g_der}
g_\alpha^{(k)}(0) = \frac{(-\alpha)^k}{k+2}.
\end{equation}
Thus we obtain that the Wronskian at zero is given by 
\begin{equation}
\label{eq:g_Wronskian}
W(g_{\alpha_k}, \dotsc, g_{\alpha_1})(0) = (k+1)!^{-k} D\begin{pmatrix}1, x, x^2 \dotsc, x^{k-1}\\ -\alpha_k, \dotsc, -\alpha_{1}\end{pmatrix}.
\end{equation}
The latter is a Vandermonde determinant, which evaluates to $\prod_{j=1}^{k-1}(\alpha_{j+1} - \alpha_j)$ and is therefore strictly positive.
\end{proof}

\subsection{Further results on interpolation polynomials}\label{app:interpolation2}

\begin{lem}\label{lem:coef_ratio}
Let $\alpha_n > \dotsc > \alpha_1 \ge 0$ be given and consider the Descartes system
\[\cD = (\varphi_{\alpha_n}, \dotsc, \varphi_{\alpha_1}), \quad \text{where} \quad \varphi_\alpha(x) = e^{-\alpha x}.\]
Let $f(x,\bm{r}) = \sum_{i=1}^n a_i(\bm{r}) \varphi_{\alpha_i}(x)$ be the interpolation D-polynomial \eqref{eq:interpol} of $\bm{r} \in \Delta_{n-1}$. Then
its coefficients satisfy, for any $i, j \in \set{1, \dotsc, n}$,
\begin{equation}\label{eq:coef_asymp}
\lim_{\bm{r} \to \bm{0}} \frac{a_i(\bm{r})}{a_j(\bm{r})}  = (-1)^{(i-j)}\frac{\alpha_{i+1} - \alpha_{i-1}}{(\alpha_{i+1} - \alpha_{i})(\alpha_{i} - \alpha_{i-1})} \frac{(\alpha_{j+1} - \alpha_{j})(\alpha_{j} - \alpha_{j-1})} {\alpha_{j+1} - \alpha_{j-1}}
\end{equation}
with the convention that terms containing $\alpha_0$ or $\alpha_{n+1}$ shall be omitted. The same result holds for $\cD$ replaced with 
\[\cE = (g_{\alpha_n}, \dotsc, g_{\alpha_1}), \quad \text{where} \quad g_\alpha(x) = \frac{1}{x^2}\int_0^x y e^{-\alpha y} dy.\]
\end{lem}
\begin{proof}
Combining \eqref{eq:coefficients} with \cite[Ch.~6, Eqs.(1.3), (1.4)]{karlin1968total}, we obtain
\begin{align*}
\lim_{\bm{r} \to \bm{0}} \frac{a_i(\bm{r})}{a_j(\bm{r})} &= (-1)^{(i-j)} \lim_{\bm{r} \to \bm{0}}  \frac{
D\begin{pmatrix}\varphi_n, \dotsc, \varphi_{i+1}, \; \varphi_{i-1}, \dotsc, \varphi_1 \\ r_1, \dotsc, r_{n-1}\end{pmatrix}
}{
D\begin{pmatrix}\varphi_n, \dotsc, \varphi_{j+1}, \; \varphi_{j-1}, \dotsc, \varphi_1 \\ r_1, \dotsc, r_{n-1}\end{pmatrix}
} = \\
&= (-1)^{(i-j)} \frac{
W\left(\varphi_n, \dotsc, \varphi_{i+1}, \; \varphi_{i-1}, \dotsc, \varphi_1 \right)(0)
}{
W\left(\varphi_n, \dotsc, \varphi_{j+1}, \; \varphi_{j-1}, \dotsc, \varphi_1 \right)(0).
} 
\end{align*}
As $\varphi_\alpha(x) = e^{-\alpha x}$, the Wronskian determinants become Vandermonde determinants, i.e.
\begin{align*}
&W\left(\varphi_n, \dotsc, \varphi_{i+1}, \; \varphi_{i-1}, \dotsc, \varphi_1 \right)(0) = D\begin{pmatrix}1, x ,x^2, \dotsc, x^{n-1}\\ -\alpha_n, \dotsc, -\alpha_{i+1}, \; -\alpha_{i-1}, \dotsc, -\alpha_1 \end{pmatrix} = \\ &\qquad \frac{\alpha_{i+1} - \alpha_{i-1}}{(\alpha_{i+1} - \alpha_{i})(\alpha_{i} - \alpha_{i-1})} \prod_{k=1}^{n-1} (\alpha_k - \alpha_{k-1}),
\end{align*}
and similarly for $j$. Evaluating their ratio, \eqref{eq:coef_asymp} is obtained. For $g$ the proof is analogous, using \eqref{eq:g_Wronskian} to evaluate the Wronskians.
\end{proof}

\begin{lem}\label{lem:special}
Consider the Descartes system $\cD_{4,\text{prox}}  = (\varphi_{2\lambda_2}, \varphi_{\lambda_2 + \lambda_1}, \varphi_{2\lambda_1}, \varphi_{\lambda_2})$ on $[0,\infty)$. There exists a neighborhood $N$ of $\bm{0}$ in $[0,\infty)^3$, such that the coefficients of the interpolation D-polynomial 
\[f(x;\bm{r}) = a_{2\lambda_2}(\bm{r})\varphi_{2\lambda_2}(x) +  a_{\lambda_1 + \lambda_2}(\bm{r})\varphi_{\lambda_1 + \lambda_2}(x) + a_{2\lambda_1}(\bm{r})\varphi_{2\lambda_1}(x) + a_{\lambda_2}(\bm{r})\varphi_{\lambda_2}(x)\]
satisfy
\begin{equation}\label{eq:coef_ineq}
\left| \frac{a_{\lambda_1 + \lambda_2}(\bm{r})}{\sqrt{a_{2 \lambda_1}(\bm{r}) a_{2 \lambda_2}(\bm{r})}}\right| < 2 \qquad \forall \,\bm{r} \in N \cap \Delta_3[0,\infty).
\end{equation}
The same holds for $\cD, \cE_{4,\text{prox}}$ and $\cE$.
\end{lem}
\begin{proof}
Applying Lemma~\ref{lem:coef_ratio} to $\cD_{4,\text{prox}}$, we calculate the limits
\begin{align*}
\lim_{\bm{r} \to \bm{0}} \left|\frac{a_{\lambda_1 + \lambda_2}(\bm{r})}{a_{\lambda_1}(\bm{r})}\right| &= 2 \left(2 - \frac{\lambda_2}{\lambda_1}\right)\\
\lim_{\bm{r} \to \bm{0}} \left|\frac{a_{\lambda_1 + \lambda_2}(\bm{r})}{a_{\lambda_2}(\bm{r})}\right| &= 2
\end{align*}
Taking square roots and multiplying, we obtain
\[\lim_{\bm{r} \to \bm{0}}\left| \frac{a_{\lambda_1 + \lambda_2}(\bm{r})}{\sqrt{a_{2 \lambda_1}(\bm{r}) a_{2 \lambda_2}(\bm{r})}}\right| = 2 \sqrt{2 - \frac{\lambda_2}{\lambda_1}}.\]
As $\lambda_1 < \lambda_2 < 2 \lambda_1$, the right hand side is contained strictly between $0$ and $2$. Due to \eqref{eq:coefficients}, the coefficients of the interpolation D-polynomial $f(x;\bm{r})$ depend continuously on $\bm{r} \in \Delta_3[0,\infty)$, and \eqref{lem:coef_ratio} follows. The proof for $\cD, \cE_{4,\text{prox}}$ and $\cE$ is analogous.
\end{proof}

%

\bibliographystyle{alpha}
\bibliography{references}

\begin{thebibliography}{CIJR85}

\bibitem[And87]{ando1987totally}
Tsuyoshi Ando.
\newblock Totally positive matrices.
\newblock {\em Linear algebra and its applications}, 90:165--219, 1987.

\bibitem[BE95]{borwein1995polynomials}
Peter Borwein and Tam{\'a}s Erd{\'e}lyi.
\newblock {\em Polynomials and polynomial inequalities}, volume 161.
\newblock Springer Science \& Business Media, 1995.

\bibitem[BM07]{brigo2007interest}
Damiano Brigo and Fabio Mercurio.
\newblock {\em Interest rate models-theory and practice: with smile, inflation
  and credit}.
\newblock Springer Science \& Business Media, 2007.

\bibitem[CIJR85]{cox1985theory}
John~C Cox, Jonathan~E Ingersoll~Jr, and Stephen~A Ross.
\newblock A theory of the term structure of interest rates.
\newblock {\em Econometrica}, 53(2):385--408, 1985.

\bibitem[DK19]{diez2019yield}
Franziska Diez and Ralf Korn.
\newblock Yield curve shapes of vasicek interest rate models, measure
  transformations and an application for the simulation of pension products.
\newblock {\em European Actuarial Journal}, pages 1--30, 2019.

\bibitem[DS00]{dai2000specification}
Qiang Dai and Kenneth~J Singleton.
\newblock Specification analysis of affine term structure models.
\newblock {\em The Journal of Finance}, 55(5):1943--1978, 2000.

\bibitem[Hog13]{hogben2013handbook}
Leslie Hogben.
\newblock {\em Handbook of linear algebra}.
\newblock Chapman and Hall/CRC, 2013.

\bibitem[HW90]{hull1990pricing}
John Hull and Alan White.
\newblock Pricing interest-rate-derivative securities.
\newblock {\em The review of financial studies}, 3(4):573--592, 1990.

\bibitem[Kar68]{karlin1968total}
Samuel Karlin.
\newblock {\em Total positivity}, volume~1.
\newblock Stanford University Press, 1968.

\bibitem[Kij02]{kijima2002monotonicity}
Masaaki Kijima.
\newblock Monotonicity and convexity of option prices revisited.
\newblock {\em Mathematical Finance}, 12(4):411--425, 2002.

\bibitem[KK13]{korn2013optionsbewertung}
Ralf Korn and Elke Korn.
\newblock {\em Optionsbewertung und Portfolio-Optimierung: Moderne Methoden der
  Finanzmathematik}.
\newblock Springer-Verlag, 2013.

\bibitem[KR18]{keller-ressel2018correction}
Martin Keller-Ressel.
\newblock Correction to: Yield curve shapes and the asymptotic short rate
  distribution in affine one-factor models.
\newblock {\em Finance and Stochastics}, 22(2):503--510, 2018.

\bibitem[KRS08]{keller-ressel2008yield}
Martin Keller-Ressel and Thomas Steiner.
\newblock Yield curve shapes and the asymptotic short rate distribution in
  affine one-factor models.
\newblock {\em Finance and Stochastics}, 12(2):149 -- 172, 2008.

\bibitem[KS66]{karlin1966tchebycheff}
Samuel Karlin and William~J Studden.
\newblock {\em Tchebycheff systems: with applications in analysis and
  statistics}.
\newblock Interscience, 1966.

\bibitem[LP07]{lord2007level}
Roger Lord and Antoon Pelsser.
\newblock Level--slope--curvature--fact or artefact?
\newblock {\em Applied Mathematical Finance}, 14(2):105--130, 2007.

\bibitem[SS06]{salinelli2006correlation}
Ernesto Salinelli and Carlo Sgarra.
\newblock Correlation matrices of yields and total positivity.
\newblock {\em Linear algebra and its applications}, 418(2-3):682--692, 2006.

\bibitem[Vas77]{vasicek1977equilibrium}
Oldrich Vasi{\v{c}}ek.
\newblock An equilibrium characterization of the term structure.
\newblock {\em Journal of Financial Economics}, 5:177--188, 1977.

\end{thebibliography}
\end{document}